\documentclass[aps,prl,groupedaddress,nofootinbib,notitlepage,showpacs,floatfix,twocolumn]{revtex4-1}

\usepackage{graphicx,graphics,epsfig,subfigure,times,bm,bbm,amssymb,amsmath,amsfonts,mathrsfs}
\usepackage[matrix,frame,arrow]{xypic}
\usepackage[pdfstartview=FitH]{hyperref}
\usepackage[pdftex]{color}

\newcommand{\beq}{\begin{equation}}
\newcommand{\eeq}{\end{equation}}
\newcommand{\beqnn}{\begin{equation*}}
\newcommand{\eneqnn}{\end{equation*}}
\newcommand{\beqy}{\begin{eqnarray}}
\newcommand{\eneqy}{\end{eqnarray}}
\newcommand{\beqynn}{\begin{eqnarray*}}
\newcommand{\eneqynn}{\end{eqnarray*}}

\newcommand{\ket}[1]{ | #1\rangle}

\newcommand{\bO}{{\bf \Omega}}
\newcommand{\hbO}{\hat{\bf {\Omega}}}

\newcommand{\bes} {\begin{subequations}}
\newcommand{\ees} {\end{subequations}}
\newcommand{\bea} {\begin{eqnarray}}
\newcommand{\eea} {\end{eqnarray}}

\newenvironment{proof}[1][Proof]{\noindent\textbf{#1.} }{\ \rule{0.5em}{0.5em}}
\newtheorem{mytheorem}{Theorem}

\newtheorem{mycorollary}{Corollary}

\newcommand{\ignore}[1]{}

\begin{document}

\title{Optimally combining dynamical decoupling and quantum error correction}
\author{Gerardo~A. Paz-Silva$^{(1,4)}$ and D. A. Lidar$^{(1,2,3,4)}$}

\begin{abstract}
We show how dynamical decoupling (DD) and quantum error correction (QEC) can be optimally combined in the setting of fault tolerant quantum computing. To this end we identify the optimal generator set of DD sequences designed to protect quantum information encoded into stabilizer subspace or subsystem codes. This generator set, comprising the stabilizers and logical operators of the code, minimizes a natural cost function associated with the length of DD sequences. We prove that with the optimal generator set the restrictive local-bath assumption used in earlier work on hybrid DD-QEC schemes, can be significantly relaxed, thus bringing hybrid DD-QEC schemes, and their potentially considerable advantages, closer to realization.

\end{abstract}

\affiliation{Departments of $^{(1)}$Chemistry, $^{(2)} $Physics, and $^{(3)}$Electrical
Engineering, and $^{(4)}$Center for Quantum Information Science \&
Technology, University of Southern California, Los Angeles, California
90089, USA}
\maketitle

{\it Introduction}.---%
The nemesis of quantum information processing is decoherence, the outcome of the inevitable interaction of a quantum system with its environment, or bath. Several methods exist that are capable of mitigating this undesired effect. Of particular interest to us here are quantum error correction (QEC) \cite{Shor:95,Calderbank:96,Steane:96b,Gottesman:96} and dynamical decoupling (DD) \cite{Viola:98,Zanardi:98b,Duan:98e,Viola:99,Yang:2010:2}. 
QEC is a closed-loop control scheme which encodes information and flushes entropy from the system via a continual supply of fresh ancilla qubits, which carry off error syndromes. DD is an open-loop control scheme that reduces the rate of entropy growth by means of pulses applied to the system, which stroboscopically decouple it from the environment. 
QEC and DD have complementary strengths and weaknesses. QEC is relatively resource-heavy, but can be extended into a fully fault-tolerant scheme, complete with an accuracy threshold theorem \cite
{Aharonov:08,Knill:97a,Terhal:04,Aliferis:05,Aharonov:05,ng:032318}. DD demands significantly more modest resources, can theoretically achieve arbitrarily high decoherence suppression \cite{KhodjastehLidar:04,Uhrig:07,Yang:2008:180403,CUDD,WFL:09,NUDD,Kuo:1106.2151,PhysRevA.84.060302}, but cannot by itself be made fully fault-tolerant \cite{PhysRevA.83.020305}.  

A natural question is whether a hybrid QEC-DD scheme is advantageous
relative to using each method separately in the setting of fault-tolerant quantum computing (FTQC). {Typically, improvements in gate accuracy achieved by DD
mean that more noise can be tolerated by a hybrid QEC-DD scheme than by QEC
alone, and that invoking DD can reduce the overhead cost of QEC. While early studies identified various advantages \cite{ByrdLidar:01a,Boulant:02,KhodjastehLidar:03}, they did not address fault tolerance. A substantial step forward was taken in
Ref.~\cite{NLP:09}, which analyzed \textquotedblleft DD-protected
gates\textquotedblright\ (DDPGs) in the FTQC 
setting. Such gates are obtained by preceding every physical gate (i.e., a
gate acting directly on the physical qubits) in a fault tolerant quantum
circuit by a DD sequence. DDPGs can be less noisy than the bare,
unprotected gates, since DD sequences can substantially reduce the strength
of the effective system-environment interaction just at the moment before
the physical gate is applied. 
%Ref.~\cite{NLP:09} derived a \textquotedblleft noise-suppression threshold condition\textquotedblright\ on the noise parameters. When this condition is satisfied, DDPGs are more accurate than unprotected gates. Typically, improvements in gate accuracy achieved by DD mean that more noise can be tolerated by a hybrid QEC-DD\ scheme than by QEC alone, and that invoking DD can reduce the overhead cost of QEC. 
The gains can be very substantial if the intrinsic noise per gate is sufficiently small, and can make quantum computing scalable with DDPGs, where it was not with unprotected gates \cite{NLP:09}.} 

The analysis in Ref.~\cite{NLP:09} assumed a ``local" perspective. Rather than analyzing the complete FT quantum circuit, each single- or multi-qubit gate was separately DD-protected. This required a strong locality
constraint limiting the spatial correlations in the noise, known as the
\textquotedblleft local bath\textquotedblright\ assumption. 
Unfortunately, many physically relevant error models violate this assumption~\cite{Aharonov:05,Aliferis:05,ng:032318}. 

Here we aim to integrate DD with FTQC using a \emph{global} perspective. This appears to be necessary in order to achieve high order decoupling in a
multi-qubit setting, under general noise models. Rather than protecting individual gates we shall show how an entire FT quantum register, including data and ancilla qubits, can be enhanced using DD. {This will allow us to relax}
%Since we deal with the entire register at once, we do not require 
the restrictive local bath assumption.
%Instead, our noise model will be compatible with the general analysis of non-Markovian FTQC~\cite{Aharonov:05,Aliferis:05,ng:032318}. 
%In achieving this, our main contributions are the 
{Along the way, we identify}
%identification of 
a DD strategy that takes into account the basic structure and building blocks of FT quantum circuits, and 
%the identification of 
{identify optimal}
DD pulse sequences compatible with this structure, that drastically reduce the number of pulses required compared with previous designs. Such a reduction is crucial in order to reap the benefits of DD protection, for if a DD sequence becomes too long, noise can accumulate to such an extent as to outweigh any DD enhancements.

{\it The noise model}.---%
We assume a completely general noise Hamiltonian $H$ acting on the joint system-bath Hilbert space, the only assumption being that $\|H\| < \infty$, where $\|\cdot \|$ denotes the sup-operator norm (the largest singular value, or largest eigenvalue for positive operators) \cite{norm-to-freq}. Informally,  $H$ contains a ``good" and a ``bad" part, the latter being the one we wish to decouple. $H$ is $k$-local, i.e., involves up to $k$-body interactions, with $k\geq 1$. We allow for arbitrary interactions between the system and the bath, as well as between different parts of the system or between different parts of the bath. See Fig.~\ref{fig:1}.
\begin{figure}[htbp]
\centering
\includegraphics[width=\columnwidth]{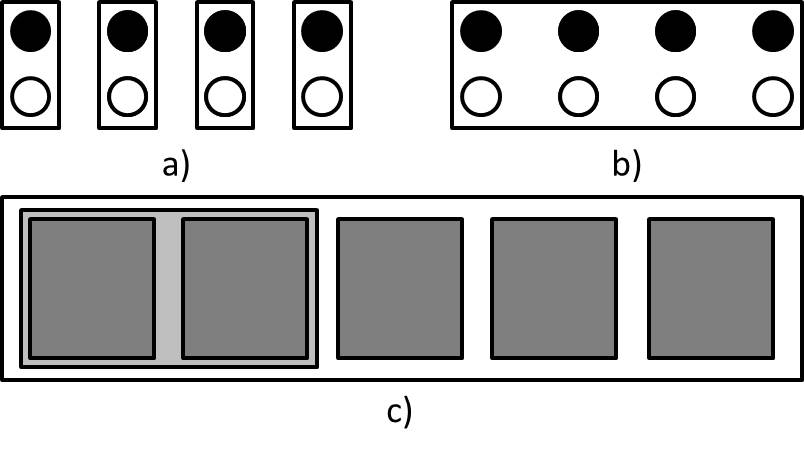}
\caption{Qubits and corresponding baths represented as white and black circles respectively. Bath operators corresponding to different operators inside a box do not necessarily commute, while they do if the baths are in different boxes. The Hamiltonians considered are general within each box, but not between them. In (a) a diagram of the ``local bath assumption" used in Ref.~\cite{NLP:09} is shown, while (b) represents the general scenario considered in fault-tolerance~\cite{Aharonov:05,Aliferis:05,ng:032318}. In (c) we illustrate one of our key results: domains are allowed to grow logarithmically in the size of the problem the FTQC is solving. The dark grey boxes represent such domains, each containing $O[\log(k_{\textrm{tot}})]$ physical qubits at the highest level of concatenation, where $k_{\textrm{tot}}$ is the total number of logical qubits. When two domains need to interact (light grey box), then the joint DD generator set is used and the locality of the bath is updated accordingly.}
\label{fig:1}
\end{figure}

{\it Dynamical decoupling}.---% 
DD pulse sequences comprise a series of rapid unitary rotations of the system qubits about different axes, separated by certain pulse intervals, and generated by a control Hamiltonian $H_C(t)$.
They are designed to suppress decoherence arising from the ``bad" terms in $H$. This is typically manifested in the suppression or even
vanishing of the first $N$ orders, in powers of the total evolution time $T$, of the time-dependent perturbation
expansion (Dyson or Magnus series \cite{Blanes:08}) of the evolution operator $U(T) = \mathcal{T}\exp(-i\int_0^T H(t)dt)$, where $H(t)$ is $H$ in the ``toggling frame" (the interaction picture generated by the DD pulse sequence Hamiltonian $H_C(t)$) \cite{Viola:99}, and $\mathcal{T}$ denotes time-ordering. When the first non-identity system-term of $U(T)$ appears at $O(T^{N+1})$ one speaks of \textit{$N$th order decoupling}. Such DD sequences are now known and well understood \cite{DD-review}. 

Most DD sequences 
can be defined in terms of  pulses chosen from a mutually orthogonal operator set (MOOS), i.e., a set of unitary and Hermitian operators $\bO =\{ \Omega_i \}_{i=1}^{|\bO|}$, $\Omega_i^2=\openone$ (identity) $\forall i$, and such that any pair of operators either commute or anticommute \cite{NUDD}. The generator set of a MOOS (gMOOS), $\hbO= \{ \Omega_i \}_{i=1}^{|\hbO|}$, is defined as the minimal subset $\hbO \subseteq \bO$ such that every element of $\bO$ is a product of elements of $\hbO$ but no element in $\hbO$ is itself a product of elements in $\hbO$ \cite{notation}.
All deterministic DD sequences are finitely generated, meaning that the pulses are elements, or products of elements, of a finite DD generator set (DDGS), \emph{which we identify with the gMOOS $\hbO$}.

The centralizer of the MOOS $\bO$ is ${\mathcal{C}}_{\bO} := \{A\ |\ [A,\bO]=0\}$, i.e., the set of operators which commute with all MOOS elements.
A good example of a gMOOS is the generator set $\hat{P}_n=\{X^{(i)},Z^{(i)}\}_{i=1}^n$, where $X^{(i)}$ ($Z^{(i)}$) denotes the Pauli-$x$ ($z$) matrix acting on the $i$th qubit, of the Pauli group $P_n = P_1^{\otimes n}$ on $n$ qubits (the group of all $n$-fold tensor products of the standard Pauli matrices $P_1 = \{\openone,X,Y,Z\}$, modulo $\mathbb{Z}_2$). 
For simplicity, since we will be dealing with qubits and are particularly interested in decoupling sequences that allow for bitwise pulses, we shall assume henceforth that $\hbO \subseteq {P}_n$. It is necessary to recast the notion of decoupling order in the MOOS scenario, since the previously mentioned notion turns out to be too strong for our purposes.

Note that any operator $A$ can decomposed as  $A= A_{0} + A_{r}$, where $A_0$ ($A_r$) denotes the component that commutes (does not commute) with all elements of a MOOS, i.e., $A_0 \in {\mathcal{C}}_{\bO}$. We shall say that a pulse sequence with generator set $\hbO$ lasting total time $T$ achieves ``$N$th order $\hbO$-decoupling'' if the joint system-bath unitary evolution operator at the conclusion of the sequence  becomes
\begin{equation}
\label{sxdd}
U^{[N]} (T) = e^{iT H^{\textrm{eff},N} (T)}  ,
\end{equation}
where the effective Hamiltonian is
\bes
\bea
H^{\textrm{eff},N} (T) &\equiv& H^{\textrm{eff},N}_{0} (T)   + H^{\textrm{eff},N}_r (T)\\
T\|H^{\textrm{eff},N}_r(T)\| &\sim& O[(T \Vert H \Vert )^{N+1}]\\
\label{eq:2c}
H^{\textrm{eff},N}_{0} (T) &\in& {\mathcal{C}}_{\bO}.
\eea
\ees
The subspace invariant under ${\mathcal{C}}_\bO$ has therefore been decoupled, in the sense that terms not commuting with ${\mathcal{C}}_\bO$ appear only in $O(T^{N+1})$. Thus the choice of the pulse generator set $\hbO$  determines what subspace(s) can be decoupled, and conversely a subspace one is interested in decoupling to arbitrary order implies a choice of $\hbO$. 

{\it Optimization of the DDGS}.---% 
We define the cost of a DD sequence as the total number of pulse intervals it uses to achieve $N$th order $\hbO$-decoupling. For all known DD sequences (even those optimized for multiple qubits \cite{PhysRevA.73.062317}), the cost is at least 
\beq
c_{\hbO}=f(N)^{|\hbO|},
\label{eq:cost}
\eeq 
and $f(N)$ depends on the particular DD sequence. Pulse interval optimization has already reduced $f(N)$ from $2^N$ for CDD to $N+1$ for NUDD \cite{DD-review}.  Here we are concerned instead with the optimization of the cost exponent $|\hbO|$, to which end the following theorem will prove to be crucial (for the proof see \cite{proof-th1}) \cite{th1-general}: 
\begin{mytheorem}
\label{gral}
Let $B$ be a subgroup of the the Pauli group $P_n$, generated by $\hat{B}$. 
Consider a DDGS $\hbO \subseteq P_n$ which decouples ${B}$ in the sense that the only element in the intersection between ${\mathcal{C}}_\bO$ and $B$ is $\openone$. Then $|\hbO | {\geq} |\hat{B}|$. Moreover, the DDGS $\hbO= \hat{B}$ decouples ${B}$ in the desired sense, and automatically saturates the bound. 
\end{mytheorem}  
As an immediate application, we reproduce the well-known result that $\hbO = \hat{P_n}$, and hence $|\hbO|=2n$, is optimal for $n$ qubits without encoding \cite{Viola:99}. Indeed, in this case the most general noise Hamiltonian is spanned by the elements of the ``error group" $B={P}_n$, so $|\hat{B}| = {2n}$ and thus by Theorem~\ref{gral} for any DDGS $\hbO$ it must be that $|\hbO| \geq 2n $. On the other hand $\hbO = \hat{P_n}$ indeed decouples ${P_n}$ since ${\mathcal{C}}_\bO = \openone$. Note also that Eq.~\eqref{eq:2c} yields $H^{\textrm{eff},N}_{0} (T) \propto \openone$. Moreover, since DD sequences are known that achieve $N$th order $\hbO$-decoupling for $n\geq 1$ qubits (specifically CDD \cite{KhodjastehLidar:04} and NUDD \cite{NUDD}, with explicit $\hat{P}_n$-based constructions given in Ref.~\cite{NUDD}), the generating set $\hat{P}_n$ is the smallest one capable of achieving $N$th order decoupling of a general $n$-qubit Hamiltonian. However, as we discuss next, there is a better choice for the purpose of protecting a code subspace.

\textit{DD generator set for a QEC code}.---% 
Consider a set of $n$ physical qubits encoding $k$ logical and $r$ gauge qubits via some distance $d$ code, i.e., an $[[n,k,r,d]]$ subsystem code \cite{Bacon:05,Kribs:2005fk,Poulin:2005uq} (or an $[[n,k,d]]$ stabilizer code \cite{Gottesman:96} for $r=0$), subject to the general noise model described above. Let $\hat{\bf S} = \{ S_\mu \}_{\mu=1}^Q$ denote the stabilizer generators, where $Q=n-(k+r)$, let $\hat{\bf L}=\{ X_L^{(i)},Z_L^{(i)}\}_{i=1}^{k}$ denote the logical-operator generators of the code, and $\hat{\bf G}=\{X_\nu,Z_\nu\}_{\nu=1}^{r}$ the gauge operator generators. In the $[[n,k,d]]$ code case, each error correctable by the code maps a codeword to an syndrome subspace labeled by an error syndrome, i.e., a sequence of $\pm 1$ eigenvalues of the stabilizer generators \cite{Gottesman:96}. In order to properly integrate DD with QEC, we require a set of DD generators $\hbO$ which preserves the error syndromes to order $N$, i.e., such that $H_{0}^{\textrm{eff},N}$ acts trivially on each of the syndrome subspaces and does not mix them, so that at the conclusion of the sequence the original noise model for which the code was chosen, is preserved (again, to order $N$). This form of the $N$th order $\hbO$-decoupling requirement will enable error correction to function as intended. A key observation is that in light of this, we do not need to protect the complete  $2^n$-dimensional Hilbert space $\mathcal{H}$, but rather the $2^{n-k}$ syndrome subspaces. 

To this end we propose to choose a complete set of stabilizer and logical operator generators as DD sequence generators, i.e., let $\hbO= \hat{\bf S}\cup \hat{\bf L}$  \cite{SDD}. We refer to any DD sequence having a DDGS of this type as an ``SLDD" sequence. With this choice, $\hat{\mathcal{C}}_\bO = \hat{\bf S} \subset \hbO$. Now note that if $\hat{\mathcal{C}}_\bO\, \subseteq \, \bO$, then the elements in $\hat{\mathcal{C}}_\bO$ commute and they define $ |{\mathcal{C}}_\bO|=2^{|\hat{\mathcal{C}}_\bO|}$ subspaces characterized by their eigenvalue under the action of $\hat{\mathcal{C}}_\bO$. In this case we have independent $N$th order $\hbO$-decoupling of each of these subspaces. In other words, $H_{{0}}^{{\rm eff},N}$ leaves each of the syndrome subspaces invariant and does not mix them, as desired. Note that the choice $\hbO =  \hat{\bf S}\cup \hat{\bf L}$ also applies to subsystem codes~\cite{Bacon:05,Kribs:2005fk,Poulin:2005uq}. In this case each of the syndrome subspaces can be decomposed as $\mathcal{H}_{\textrm{logical}} \otimes \mathcal{H}_{\textrm{gauge}}$, where $\mathcal{H}_{\textrm{logical}}$ is invariant under $ \hat{\bf S}\cup\hat{\bf G} = \mathcal{C}_{\hbO} $, since the gauge operators act non-trivially on $\mathcal{H}_{\textrm{gauge}}$ only. Before proving its optimality, we next compare the cost of the SLDD sequence to decoupling the entire Hilbert space. 

{\it Relative cost of SLDD}.---%
For an $[[n,k,d]]$ code and an SLDD sequence, the number of stabilizer generators ($n-k$) plus logical operator generators ($2k$) yields $|\hbO|= n+k$, which means that $c_{\hat{\bf S}\cup \hat{\bf L}} = f(N)^{n+k} < c_{\hat{P}_n} = f(N)^{2n}$.  Often $n\gg k$, so that $c_{\hat{\bf S}\cup \hat{\bf L}} \sim \sqrt{c_{\hat{P}_n}}$. In the case of $[[n,k,r,d]]$ subsystem codes~\cite{Poulin:2005uq} the advantage is more pronounced: the number of stabilizers is $n-k-r$, so $|\hbO| = n+k-r$. As an example consider the Bacon-Shor $[[ m\times m, 1, (m-1)^2, 3]]$ subsystem code \cite{Bacon:05}, which has the highest (analytically) known fault-tolerant threshold for error correction routines with~\cite{Aliferis:05} and without measurements~\cite{Paz-Silva:2010fk}. In this case one would have $c_{\hat{\bf S}\cup \hat{\bf L}} = f(N)^{2 m} =  (c_{\hat{P}_{m^2}})^{1/m}$, a polynomial advantage that grows with the block size $m$.

{\it Choice of DDGS for protecting ancilla states}.---% 
The protection of certain ancilla states is also an important part of fault tolerance. Such states can be thought of as QEC codes with small stabilizer sets. E.g., $ \ket{\textrm{cat}_{m,+}} = (\ket{0}^{\otimes a} + \ket{1}^{\otimes a})/\sqrt{2}$ is often used for fault-tolerant stabilizer measurements or for teleportation of encoded information. The stabilizer is generated by $ \{ X^{\otimes a}, \{Z_iZ_{i+1}\}_{i=1}^{a-1}\}$, and equals the DDGS.

\textit{Decoupling multiple subspaces or subsystems}.---%
How should one choose an optimal DDGS to decouple different subspaces simultaneously? Assume that there are distinct and non-overlapping sets of $\{n_i\}$ physical qubits comprising a quantum register, e.g., a complete register comprising $k$ logical qubits, along with the corresponding ancillas. Assume that they are partitioned into sets of sizes $\{k_i\}_{i=1}^p$, such that $k=\sum_{i=1}^p k_i $, and that each set $i$ is encoded in some subsystem (or subspace) code $[[n_i, k_i,r_i,d_i]]$. For each block of $k_i$ logical qubits we have an SLDD sequence with DDGS $\hbO_i = \hat{\bf S}_i\cup \hat{\bf L}_i $. 
Let the Hamiltonians of the different sets be $\{H_i\}$, and spanned by the error groups $\{B_i \subset P_{n_i}\}$. Using Theorem~\ref{gral}, it follows that if $\hbO_i$ is optimal for error group $B_i$ then $\hbO_{\textrm{tot}} = \cup_i \hbO_{B_i}$ optimally decouples the joint Hamiltonian spanned by $\cup_i B_i$. This form of composing a larger DDGS out of smaller modules guarantees that each term of a general Hamiltonian acting on the whole register must anticommute with at least one element in $\hbO_{\textrm{tot}}$, which in turn implies that $\hbO_{\textrm{tot}}$, used to construct, e.g., a CDD or NUDD sequence, is capable of independent $N$th order $\hbO_{\textrm{tot}}$-decoupling of each subspace or subsystem.

{\it Optimal DDGS for concatenated QEC codes}.---%
Many FTQC constructions are based on concatenated QEC codes \cite{Gaitan:book}, so what is the optimal DDGS for this case, cost-wise? Suppose an $[[n,k,r,d]]$ code is concatenated $R$ times. A complete generator set for \emph{all} the stabilizers of such a code is given by $ \cup_{q=1}^R \hat{\bf S}^{(q)}$, where $\hat{\bf S}^{(q)}$ is the stabilizer generator set of concatenation level $q$. Let $\hat{\bf L}^{(R)}$ denote the set of $R$th-concatenation level logical generators. 
\begin{mytheorem}
\label{concat}
The optimal DDGS for decoupling all the syndrome subspaces at concatenation level $R$ is the SLDD set $\hbO= \cup_{q=1}^R \hat{\bf S}^{(q)}  \cup \hat{\bf L}^{(R)}$, where $|\hbO|=n^R-(k+r)^R+2k^R$.
\end{mytheorem}
Note that there are alternatives to this ``top-level" SLDD strategy; e.g., one could concatenate the DDGS for each block at each level $q$, but this would result in exponentially more DD pulses. Note also that by setting $R=1$ Theorem~\ref{concat} reduces to the optimality of SLDD for subspace or subsystem codes, with $|\hbO|=n+k-r$ as claimed above. The subspace case is recovered by setting $r=0$.

\begin{proof}
The number of physical qubits after $R$ levels of concatenation of any $[[n,k,r,d]]$ subsystem stabilizer code is $n(R)=n^R$, and the error group for the entire Hilbert space is the Pauli group $P_{n(R)}$. We need to protect the $2^{Q(R)}$ syndrome subspaces, where $Q(R)=|\cup_{q=1}^R \hat{\bf S}^{(q)}|$ is the total number of stabilizer generators after the code is concatenated $R$ times. 
$Q(R) = n(R)-L(R)-G(R)$, where $L(R)=k^R$ [$G(R)$] is the number of logical (gauge) qubits at level $R$, and $L(R)+G(R) = (k+r)^R$ \cite{stabilizer-count}.

The SLDD sequence generated by $\hbO= \cup_{q=1}^R \hat{\bf S}^{(q)}  \cup \hat{\bf L}^{(R)}$ satisfies the requirement of independent $N$th order $\hbO$-decoupling  of the $2^{Q(R)}$ syndrome subspaces since the stabilizers (as DD pulses) remove the errors at each level $q$, logical included (recall that a logical error at level $q-1$ anticommutes with at least one level $q$ stabilizer generator), but not the logical errors at the top level, for which we need $\hat{\bf L}^{(R)}$ as DD pulses.  Moreover, for this sequence $|\hbO| = Q(R) +2L(R) = 
%n(R)-[L(R)+G(R)]+2L(R) = 
n^R-(k+r)^R+2k^R$ as claimed. Thus what remains is to prove its optimality.

Any operator in $P_{n(R)}$ which is not a stabilizer or gauge operator acts as an error either within or between syndrome subspaces. Thus our choice of code dictates which elements of $P_{n(R)}$ act as errors, and clearly this error set is precisely $B=P_{n(R)}/\mathcal{C}_{\bO}$, where the centralizer generator is $\hat{\mathcal{C}}_{\bO} = \cup_{q=1}^R \hat{\bf S}^{(q)}  \cup \hat{\bf G}^{(q)}$. We have $|\hat{\mathcal{C}}_{\bO}| = Q(R)+2G(R)$. 
On the other hand $|\hat{B}| = 2n(R) - |\hat{\mathcal{C}}_{\bO}| = Q(R)+2L(R)$, so that $|\hat{B}| = |\hbO|$ and $B\cap\mathcal{C}_{\bO}=\openone$, which proves the optimality of $|\hbO|$ by virtue of Theorem~\ref{gral}.
\end{proof}

\textit{Optimizing the choice of DDGS for a complete quantum register: beyond the local bath assumption}.---% 
We have now assembled and described all the ingredients for optimally combining DD with FTQC for protection of a complete quantum register. However, we must ensure that the cost of implementing the DD sequence does not spoil quantum speedups. To this end we consider once more an $[[n,k,r,d]]$ subsystem code concatenated $R$ times, used to encode an entire quantum register, and divide the register into $d(R)$ domains (e.g., a code block along with ancillas) of size $k_{\textrm{D}}(R) = O(k^R)$ logical qubits, such that the total number of logical qubits in the register is $k_{\textrm{tot}} = d(R)k_{\textrm{D}}(R)$. We then optimally decouple the $i$th domain using an SLDD sequence generated by $\hbO_i = \cup_{q=1}^R \hat{\bf S}^{(q)}_i  \cup \hat{\bf L}_i^{(R)}$, $i\in\{1,\dots d(R)\}$ (where $\hat{\bf S}^{(q)}_i$ and $\hat{\bf L}_i^{(R)}$ act non-trivially only on the qubits in the domain $i$), and ask for the maximal allowed size of each domain such that the DD sequence cost scales polynomially in $k_{\textrm{tot}}$, as this will ensure that any exponential quantum speedup is retained.
\begin{mycorollary}
\label{domain}
In a fault tolerant quantum computation the maximal allowed domain size compatible with a DDGS having cost  $c_{\hbO} = f(N)^{|\hbO|} = \rm{poly}(k_{\rm{tot}})$, is $O[\log(k_{\rm{tot}})]$.
\end{mycorollary}
\begin{proof}
We assume that the total cost per domain $c_{\hbO}$ is Eq.~\eqref{eq:cost} as it captures all known DD sequences. Theorem~\ref{concat} shows that $|\hbO| = O[n^R -(k+r)^R + 2 k^R]$ (the $O$ symbol is used since we allow for the presence of ancillas in the domain). We may assume that the code has parameters such that $n\sim r \sim k$, so that  $|\hbO| = O(k^R) = O[k_{\textrm{D}}(R)]$. Now recall that $R = O[\log\log (k_{\textrm{tot}})]$ in a fault-tolerant simulation of a quantum circuit \cite{Nielsen:book}. Therefore $c_{\hbO}  = f(N)^{|\hbO|} =\rm{poly}(k_{\rm{tot}})$ requires $k_{\textrm{D}}(R) = O[\log(\rm{poly}(k_{\rm{tot}}))] = O[\log(k_{\rm{tot}})]$.
\end{proof}

Corollary~\ref{domain} means that we can relax the local bath assumption, an assumption tantamount to assuming \emph{constant} domain size $k_{\textrm{D}}\leq 2$ \cite{NLP:09}; instead we find that domains are allowed to grow logarithmically with problem size. When two domains $i$ and $j$ are required to interact, the joint DDGS $\hbO_i\cup\hbO_j$ should be used [see Fig.~\ref{fig:1}(c)]. 
If the result is that at the highest concatenation level the noise per gate has been reduced (as shown explicitly for the local bath setting in Ref.~\cite{NLP:09}), then a reduction in the number of required concatenation levels is enabled, hence reducing the overall overhead, or the effective noise threshold. 

{\it Enhanced fidelity gates via DD}.---%
So far we discussed the problem of protecting stored quantum information; what about computation? Quantum logic operations can be combined with DD, e.g., using ``decouple while compute" schemes \cite{KhodjastehLidar:08,West:10}, or  (concatenated) dynamically corrected gates [(C)DCGs] for finite-width pulses \cite{KLV:09}, or dynamically protected gates~\cite{NLP:09} in the zero-width (ideal) pulse limit. The optimal SLDD scheme introduced here is directly portable into the latter two schemes, since they use the same DD building blocks and the associated group structure. It is important to emphasize that SLDD sequences require only bitwise (i.e., transversal) pulses, and can be generated by one-local Hamiltonians, thus not altering the assumptions of the CDCG construction. More importantly, the polynomial scaling guaranteed by Corollary~\ref{domain} also applies in the quantum logic scenario, thus allowing, in principle, a fidelity improvement without sacrificing the speedup of quantum computing. 

\textit{Conclusions and outlook}.---%
All known DD sequences scale exponentially with the cardinality of their generating sets [Eq.~\eqref{eq:cost}]. In this work we identified the optimal generating set in the general context of protection of encoded information. This allowed us to show how DD and FTQC can be optimally integrated. In doing so we relaxed the local-bath assumption and showed that it can be replaced with domains growing logarithmically with problem size. Two important open problems remain: to demonstrate that DD-enhanced FTQC results in improved resource overheads and lower noise thresholds, and to identify, or rule out, multi-qubit DD sequences with sub-exponential scaling in the the cardinality of their generating sets.

\textit{Acknowledgments}.---%
Supported by the US Department of Defense and the Intelligence Advanced Research Projects Activity (IARPA) via Department of Interior National Business Center contract number D11PC20165. The U.S. Government is authorized to reproduce and distribute reprints for Governmental purposes notwithstanding any copyright annotation thereon. The views and conclusions contained herein are those of
the authors and should not be interpreted as necessarily representing the official policies or endorsements, either expressed or implied, of IARPA, DoI/NBC, or the U.S. Government.

%\bibliographystyle{apsrev4-1}
%\bibliography{bib}

\begin{thebibliography}{46}%
\makeatletter
\providecommand \@ifxundefined [1]{%
 \@ifx{#1\undefined}
}%
\providecommand \@ifnum [1]{%
 \ifnum #1\expandafter \@firstoftwo
 \else \expandafter \@secondoftwo
 \fi
}%
\providecommand \@ifx [1]{%
 \ifx #1\expandafter \@firstoftwo
 \else \expandafter \@secondoftwo
 \fi
}%
\providecommand \natexlab [1]{#1}%
\providecommand \enquote  [1]{``#1''}%
\providecommand \bibnamefont  [1]{#1}%
\providecommand \bibfnamefont [1]{#1}%
\providecommand \citenamefont [1]{#1}%
\providecommand \href@noop [0]{\@secondoftwo}%
\providecommand \href [0]{\begingroup \@sanitize@url \@href}%
\providecommand \@href[1]{\@@startlink{#1}\@@href}%
\providecommand \@@href[1]{\endgroup#1\@@endlink}%
\providecommand \@sanitize@url [0]{\catcode `\\12\catcode `\$12\catcode
  `\&12\catcode `\#12\catcode `\^12\catcode `\_12\catcode `\%12\relax}%
\providecommand \@@startlink[1]{}%
\providecommand \@@endlink[0]{}%
\providecommand \url  [0]{\begingroup\@sanitize@url \@url }%
\providecommand \@url [1]{\endgroup\@href {#1}{\urlprefix }}%
\providecommand \urlprefix  [0]{URL }%
\providecommand \Eprint [0]{\href }%
\@ifxundefined \urlstyle {%
  \providecommand \doi  [0]{\begingroup \@sanitize@url \@doi}%
  \providecommand \@doi [1]{\endgroup \@@startlink {\doibase
  #1}doi:\discretionary {}{}{}#1\@@endlink }%
}{%
  \providecommand \doi  [0]{doi:\discretionary{}{}{}\begingroup
  \urlstyle{rm}\Url }%
}%
\providecommand \doibase [0]{http://dx.doi.org/}%
\providecommand \Doi [0]{\begingroup \@sanitize@url \@Doi }%
\providecommand \@Doi  [1]{\endgroup\@@startlink{\doibase#1}\@@Doi}%
\providecommand \@@Doi [1]{#1\@@endlink}%
\providecommand \selectlanguage [0]{\@gobble}%
\providecommand \bibinfo  [0]{\@secondoftwo}%
\providecommand \bibfield  [0]{\@secondoftwo}%
\providecommand \translation [1]{[#1]}%
\providecommand \BibitemOpen [0]{}%
\providecommand \bibitemStop [0]{}%
\providecommand \bibitemNoStop [0]{.\EOS\space}%
\providecommand \EOS [0]{\spacefactor3000\relax}%
\providecommand \BibitemShut  [1]{\csname bibitem#1\endcsname}%
%</preamble>
\bibitem [{\citenamefont {Shor}(1995)}]{Shor:95}%
  \BibitemOpen
  \bibfield  {author} {\bibinfo {author} {\bibfnamefont {P.}~\bibnamefont
  {Shor}},\ }\href {http://dx.doi.org/10.1103/PhysRevA.52.R2493} {\bibfield
  {journal} {\bibinfo  {journal} {Phys. Rev. A},\ }\textbf {\bibinfo {volume}
  {52}},\ \bibinfo {pages} {R2493} (\bibinfo {year} {1995})}\BibitemShut
  {NoStop}%
\bibitem [{\citenamefont {Calderbank}\ and\ \citenamefont
  {Shor}(1996)}]{Calderbank:96}%
  \BibitemOpen
  \bibfield  {author} {\bibinfo {author} {\bibfnamefont {A.}~\bibnamefont
  {Calderbank}}\ and\ \bibinfo {author} {\bibfnamefont {P.}~\bibnamefont
  {Shor}},\ }\href {http://dx.doi.org/10.1103/PhysRevA.54.1098} {\bibfield
  {journal} {\bibinfo  {journal} {Phys. Rev. A},\ }\textbf {\bibinfo {volume}
  {54}},\ \bibinfo {pages} {1098} (\bibinfo {year} {1996})}\BibitemShut
  {NoStop}%
\bibitem [{\citenamefont {Steane}(1996)}]{Steane:96b}%
  \BibitemOpen
  \bibfield  {author} {\bibinfo {author} {\bibfnamefont {A.}~\bibnamefont
  {Steane}},\ }\href {http://dx.doi.org/10.1098/rspa.1996.0136} {\bibfield
  {journal} {\bibinfo  {journal} {Proc. R. Soc. London Ser. A},\ }\textbf
  {\bibinfo {volume} {452}},\ \bibinfo {pages} {2551} (\bibinfo {year}
  {1996})}\BibitemShut {NoStop}%
\bibitem [{\citenamefont {Gottesman}(1996)}]{Gottesman:96}%
  \BibitemOpen
  \bibfield  {author} {\bibinfo {author} {\bibfnamefont {D.}~\bibnamefont
  {Gottesman}},\ }\href {http://dx.doi.org/10.1103/PhysRevA.54.1862} {\bibfield
   {journal} {\bibinfo  {journal} {Phys. Rev. A},\ }\textbf {\bibinfo {volume}
  {54}},\ \bibinfo {pages} {1862} (\bibinfo {year} {1996})}\BibitemShut
  {NoStop}%
\bibitem [{\citenamefont {Viola}\ and\ \citenamefont {Lloyd}(1998)}]{Viola:98}%
  \BibitemOpen
  \bibfield  {author} {\bibinfo {author} {\bibfnamefont {L.}~\bibnamefont
  {Viola}}\ and\ \bibinfo {author} {\bibfnamefont {S.}~\bibnamefont {Lloyd}},\
  }\href {http://dx.doi.org/10.1103/PhysRevA.58.2733} {\bibfield  {journal}
  {\bibinfo  {journal} {Phys. Rev. A},\ }\textbf {\bibinfo {volume} {58}},\
  \bibinfo {pages} {2733} (\bibinfo {year} {1998})}\BibitemShut {NoStop}%
\bibitem [{\citenamefont {Zanardi}(1999)}]{Zanardi:98b}%
  \BibitemOpen
  \bibfield  {author} {\bibinfo {author} {\bibfnamefont {P.}~\bibnamefont
  {Zanardi}},\ }\href {http://dx.doi.org/10.1016/S0375-9601(99)00365-5}
  {\bibfield  {journal} {\bibinfo  {journal} {Phys. Lett. A},\ }\textbf
  {\bibinfo {volume} {258}},\ \bibinfo {pages} {77} (\bibinfo {year}
  {1999})}\BibitemShut {NoStop}%
\bibitem [{\citenamefont {Duan}\ and\ \citenamefont {Guo}(1999)}]{Duan:98e}%
  \BibitemOpen
  \bibfield  {author} {\bibinfo {author} {\bibfnamefont {L.-M.}\ \bibnamefont
  {Duan}}\ and\ \bibinfo {author} {\bibfnamefont {G.}~\bibnamefont {Guo}},\
  }\href {http://dx.doi.org/10.1016/S0375-9601(99)00592-7} {\bibfield
  {journal} {\bibinfo  {journal} {Phys. Lett. A},\ }\textbf {\bibinfo {volume}
  {261}},\ \bibinfo {pages} {139} (\bibinfo {year} {1999})}\BibitemShut
  {NoStop}%
\bibitem [{\citenamefont {Viola}\ \emph {et~al.}(1999)\citenamefont {Viola},
  \citenamefont {Knill},\ and\ \citenamefont {Lloyd}}]{Viola:99}%
  \BibitemOpen
  \bibfield  {author} {\bibinfo {author} {\bibfnamefont {L.}~\bibnamefont
  {Viola}}, \bibinfo {author} {\bibfnamefont {E.}~\bibnamefont {Knill}}, \ and\
  \bibinfo {author} {\bibfnamefont {S.}~\bibnamefont {Lloyd}},\ }\href
  {http://dx.doi.org/10.1103/PhysRevLett.82.2417} {\bibfield  {journal}
  {\bibinfo  {journal} {Phys. Rev. Lett.},\ }\textbf {\bibinfo {volume} {82}},\
  \bibinfo {pages} {2417} (\bibinfo {year} {1999})}\BibitemShut {NoStop}%
\bibitem [{\citenamefont {{For a recent review see W. Yang, Z.-Y. Wang, and
  R.-B. Liu}}(2011)}]{Yang:2010:2}%
  \BibitemOpen
  \bibfield  {author} {\bibinfo {author} {\bibnamefont {{For a recent review
  see W. Yang, Z.-Y. Wang, and R.-B. Liu}}},\ }\Doi {10.1007/s11467-010-0113-8}
  {\bibfield  {journal} {\bibinfo  {journal} {Front. Phys.},\ }\textbf
  {\bibinfo {volume} {6}},\ \bibinfo {pages} {2} (\bibinfo {year}
  {2011})}\BibitemShut {NoStop}%
\bibitem [{\citenamefont {Aharonov}\ and\ \citenamefont
  {Ben-Or}(2008)}]{Aharonov:08}%
  \BibitemOpen
  \bibfield  {author} {\bibinfo {author} {\bibfnamefont {D.}~\bibnamefont
  {Aharonov}}\ and\ \bibinfo {author} {\bibfnamefont {M.}~\bibnamefont
  {Ben-Or}},\ }\href {http://dx.doi.org/10.1137/S0097539799359385} {\bibfield
  {journal} {\bibinfo  {journal} {SIAM J. Comput.},\ }\textbf {\bibinfo
  {volume} {38}},\ \bibinfo {pages} {1207} (\bibinfo {year}
  {2008})}\BibitemShut {NoStop}%
\bibitem [{\citenamefont {Knill}\ \emph {et~al.}(1998)\citenamefont {Knill},
  \citenamefont {Laflamme},\ and\ \citenamefont {Zurek}}]{Knill:97a}%
  \BibitemOpen
  \bibfield  {author} {\bibinfo {author} {\bibfnamefont {E.}~\bibnamefont
  {Knill}}, \bibinfo {author} {\bibfnamefont {R.}~\bibnamefont {Laflamme}}, \
  and\ \bibinfo {author} {\bibfnamefont {W.}~\bibnamefont {Zurek}},\ }\href
  {http://www.jstor.org/stable/53171} {\bibfield  {journal} {\bibinfo
  {journal} {Proc. R. Soc. London Ser. A},\ }\textbf {\bibinfo {volume}
  {454}},\ \bibinfo {pages} {365} (\bibinfo {year} {1998})}\BibitemShut
  {NoStop}%
\bibitem [{\citenamefont {Terhal}\ and\ \citenamefont
  {Burkard}(2005)}]{Terhal:04}%
  \BibitemOpen
  \bibfield  {author} {\bibinfo {author} {\bibfnamefont {B.}~\bibnamefont
  {Terhal}}\ and\ \bibinfo {author} {\bibfnamefont {G.}~\bibnamefont
  {Burkard}},\ }\href {http://link.aps.org/doi/10.1103/PhysRevA.71.012336}
  {\bibfield  {journal} {\bibinfo  {journal} {Phys. Rev. A},\ }\textbf
  {\bibinfo {volume} {71}},\ \bibinfo {pages} {012336} (\bibinfo {year}
  {2005})}\BibitemShut {NoStop}%
\bibitem [{\citenamefont {Aliferis}\ \emph {et~al.}(2006)\citenamefont
  {Aliferis}, \citenamefont {Gottesman},\ and\ \citenamefont
  {Preskill}}]{Aliferis:05}%
  \BibitemOpen
  \bibfield  {author} {\bibinfo {author} {\bibfnamefont {P.}~\bibnamefont
  {Aliferis}}, \bibinfo {author} {\bibfnamefont {D.}~\bibnamefont {Gottesman}},
  \ and\ \bibinfo {author} {\bibfnamefont {J.}~\bibnamefont {Preskill}},\
  }\href {http://dl.acm.org/citation.cfm?id=2011665.2011666} {\bibfield
  {journal} {\bibinfo  {journal} {Quantum Inf. Comput.},\ }\textbf {\bibinfo
  {volume} {6}},\ \bibinfo {pages} {97} (\bibinfo {year} {2006})}\BibitemShut
  {NoStop}%
\bibitem [{\citenamefont {Aharonov}\ \emph {et~al.}(2006)\citenamefont
  {Aharonov}, \citenamefont {Kitaev},\ and\ \citenamefont
  {Preskill}}]{Aharonov:05}%
  \BibitemOpen
  \bibfield  {author} {\bibinfo {author} {\bibfnamefont {D.}~\bibnamefont
  {Aharonov}}, \bibinfo {author} {\bibfnamefont {A.}~\bibnamefont {Kitaev}}, \
  and\ \bibinfo {author} {\bibfnamefont {J.}~\bibnamefont {Preskill}},\ }\Doi
  {10.1103/PhysRevLett.96.050504} {\bibfield  {journal} {\bibinfo  {journal}
  {Phys. Rev. Lett.},\ }\textbf {\bibinfo {volume} {96}},\ \bibinfo {pages}
  {050504} (\bibinfo {year} {2006})}\BibitemShut {NoStop}%
\bibitem [{\citenamefont {Ng}\ and\ \citenamefont
  {Preskill}(2009)}]{ng:032318}%
  \BibitemOpen
  \bibfield  {author} {\bibinfo {author} {\bibfnamefont {H.~K.}\ \bibnamefont
  {Ng}}\ and\ \bibinfo {author} {\bibfnamefont {J.}~\bibnamefont {Preskill}},\
  }\Doi {10.1103/PhysRevA.79.032318} {\bibfield  {journal} {\bibinfo  {journal}
  {Phys. Rev. A},\ }\textbf {\bibinfo {volume} {79}},\ \bibinfo {eid} {032318}
  (\bibinfo {year} {2009})}\BibitemShut {NoStop}%
\bibitem [{\citenamefont {Khodjasteh}\ and\ \citenamefont
  {Lidar}(2005)}]{KhodjastehLidar:04}%
  \BibitemOpen
  \bibfield  {author} {\bibinfo {author} {\bibfnamefont {K.}~\bibnamefont
  {Khodjasteh}}\ and\ \bibinfo {author} {\bibfnamefont {D.~A.}\ \bibnamefont
  {Lidar}},\ }\href {http://dx.doi.org/10.1103/PhysRevLett.95.180501}
  {\bibfield  {journal} {\bibinfo  {journal} {Phys. Rev. Lett.},\ }\textbf
  {\bibinfo {volume} {95}},\ \bibinfo {pages} {180501} (\bibinfo {year}
  {2005})}\BibitemShut {NoStop}%
\bibitem [{\citenamefont {Uhrig}(2007)}]{Uhrig:07}%
  \BibitemOpen
  \bibfield  {author} {\bibinfo {author} {\bibfnamefont {G.}~\bibnamefont
  {Uhrig}},\ }\href {http://dx.doi.org/10.1103/PhysRevLett.98.100504}
  {\bibfield  {journal} {\bibinfo  {journal} {Phys. Rev. Lett.},\ }\textbf
  {\bibinfo {volume} {98}},\ \bibinfo {pages} {100504} (\bibinfo {year}
  {2007})}\BibitemShut {NoStop}%
\bibitem [{\citenamefont {Yang}\ and\ \citenamefont
  {Liu}(2008)}]{Yang:2008:180403}%
  \BibitemOpen
  \bibfield  {author} {\bibinfo {author} {\bibfnamefont {W.}~\bibnamefont
  {Yang}}\ and\ \bibinfo {author} {\bibfnamefont {R.-B.}\ \bibnamefont {Liu}},\
  }\href {http://dx.doi.org/10.1103/PhysRevLett.101.180403} {\bibfield
  {journal} {\bibinfo  {journal} {Phys. Rev. Lett.},\ }\textbf {\bibinfo
  {volume} {101}},\ \bibinfo {pages} {180403} (\bibinfo {year}
  {2008})}\BibitemShut {NoStop}%
\bibitem [{\citenamefont {Uhrig}(2009)}]{CUDD}%
  \BibitemOpen
  \bibfield  {author} {\bibinfo {author} {\bibfnamefont {G.~S.}\ \bibnamefont
  {Uhrig}},\ }\href {http://dx.doi.org/10.1103/PhysRevLett.102.120502}
  {\bibfield  {journal} {\bibinfo  {journal} {Phys. Rev. Lett.},\ }\textbf
  {\bibinfo {volume} {102}},\ \bibinfo {pages} {120502} (\bibinfo {year}
  {2009})}\BibitemShut {NoStop}%
\bibitem [{\citenamefont {West}\ \emph
  {et~al.}(2010){\natexlab{a}}\citenamefont {West}, \citenamefont {Fong},\ and\
  \citenamefont {Lidar}}]{WFL:09}%
  \BibitemOpen
  \bibfield  {author} {\bibinfo {author} {\bibfnamefont {J.~R.}\ \bibnamefont
  {West}}, \bibinfo {author} {\bibfnamefont {B.~H.}\ \bibnamefont {Fong}}, \
  and\ \bibinfo {author} {\bibfnamefont {D.~A.}\ \bibnamefont {Lidar}},\ }\href
  {http://dx.doi.org/10.1103/PhysRevLett.104.130501} {\bibfield  {journal}
  {\bibinfo  {journal} {Phys. Rev. Lett.},\ }\textbf {\bibinfo {volume}
  {104}},\ \bibinfo {pages} {130501} (\bibinfo {year}
  {2010}{\natexlab{a}})}\BibitemShut {NoStop}%
\bibitem [{\citenamefont {Wang}\ and\ \citenamefont {Liu}(2011)}]{NUDD}%
  \BibitemOpen
  \bibfield  {author} {\bibinfo {author} {\bibfnamefont {Z.-Y.}\ \bibnamefont
  {Wang}}\ and\ \bibinfo {author} {\bibfnamefont {R.-B.}\ \bibnamefont {Liu}},\
  }\Doi {10.1103/PhysRevA.83.022306} {\bibfield  {journal} {\bibinfo  {journal}
  {Phys. Rev. A},\ }\textbf {\bibinfo {volume} {83}},\ \bibinfo {pages}
  {022306} (\bibinfo {year} {2011})}\BibitemShut {NoStop}%
\bibitem [{\citenamefont {Kuo}\ and\ \citenamefont
  {Lidar}(2011)}]{Kuo:1106.2151}%
  \BibitemOpen
  \bibfield  {author} {\bibinfo {author} {\bibfnamefont {W.-J.}\ \bibnamefont
  {Kuo}}\ and\ \bibinfo {author} {\bibfnamefont {D.~A.}\ \bibnamefont
  {Lidar}},\ }\Doi {10.1103/PhysRevA.84.042329} {\bibfield  {journal} {\bibinfo
   {journal} {Phys. Rev. A},\ }\textbf {\bibinfo {volume} {84}},\ \bibinfo
  {pages} {042329} (\bibinfo {year} {2011})}\BibitemShut {NoStop}%
\bibitem [{\citenamefont {Jiang}\ and\ \citenamefont
  {Imambekov}(2011)}]{PhysRevA.84.060302}%
  \BibitemOpen
  \bibfield  {author} {\bibinfo {author} {\bibfnamefont {L.}~\bibnamefont
  {Jiang}}\ and\ \bibinfo {author} {\bibfnamefont {A.}~\bibnamefont
  {Imambekov}},\ }\Doi {10.1103/PhysRevA.84.060302} {\bibfield  {journal}
  {\bibinfo  {journal} {Phys. Rev. A},\ }\textbf {\bibinfo {volume} {84}},\
  \bibinfo {pages} {060302} (\bibinfo {year} {2011})}\BibitemShut {NoStop}%
\bibitem [{\citenamefont {Khodjasteh}\ \emph {et~al.}(2011)\citenamefont
  {Khodjasteh}, \citenamefont {Erd\'elyi},\ and\ \citenamefont
  {Viola}}]{PhysRevA.83.020305}%
  \BibitemOpen
  \bibfield  {author} {\bibinfo {author} {\bibfnamefont {K.}~\bibnamefont
  {Khodjasteh}}, \bibinfo {author} {\bibfnamefont {T.}~\bibnamefont
  {Erd\'elyi}}, \ and\ \bibinfo {author} {\bibfnamefont {L.}~\bibnamefont
  {Viola}},\ }\href {http://dx.doi.org/10.1103/PhysRevA.83.020305} {\bibfield
  {journal} {\bibinfo  {journal} {Phys. Rev. A},\ }\textbf {\bibinfo {volume}
  {83}},\ \bibinfo {pages} {020305} (\bibinfo {year} {2011})}\BibitemShut
  {NoStop}%
\bibitem [{\citenamefont {Byrd}\ and\ \citenamefont
  {Lidar}(2002)}]{ByrdLidar:01a}%
  \BibitemOpen
  \bibfield  {author} {\bibinfo {author} {\bibfnamefont {M.~S.}\ \bibnamefont
  {Byrd}}\ and\ \bibinfo {author} {\bibfnamefont {D.~A.}\ \bibnamefont
  {Lidar}},\ }\href {http://link.aps.org/doi/10.1103/PhysRevLett.89.047901}
  {\bibfield  {journal} {\bibinfo  {journal} {Phys. Rev. Lett.},\ }\textbf
  {\bibinfo {volume} {89}},\ \bibinfo {pages} {047901} (\bibinfo {year}
  {2002})}\BibitemShut {NoStop}%
\bibitem [{\citenamefont {{N. Boulant, M.A. Pravia, E.M. Fortunato, T.F. Havel
  and D.G. Cory}}(2002)}]{Boulant:02}%
  \BibitemOpen
  \bibfield  {author} {\bibinfo {author} {\bibnamefont {{N. Boulant, M.A.
  Pravia, E.M. Fortunato, T.F. Havel and D.G. Cory}}},\ }\href
  {http://dx.doi.org/10.1023/A:1019623208633} {\bibfield  {journal} {\bibinfo
  {journal} {Quant. Inf. Proc.},\ }\textbf {\bibinfo {volume} {1}},\ \bibinfo
  {pages} {135} (\bibinfo {year} {2002})}\BibitemShut {NoStop}%
\bibitem [{\citenamefont {{K. Khodjasteh and D. A.
  Lidar}}(2003)}]{KhodjastehLidar:03}%
  \BibitemOpen
  \bibfield  {author} {\bibinfo {author} {\bibnamefont {{K. Khodjasteh and D.
  A. Lidar}}},\ }\href {http://dx.doi.org/10.1103/PhysRevA.68.022322}
  {\bibfield  {journal} {\bibinfo  {journal} {Phys. Rev. A},\ }\textbf
  {\bibinfo {volume} {68}},\ \bibinfo {pages} {022322} (\bibinfo {year}
  {2003})},\ \bibinfo {note} {erratum: {\it ibid}, Phys. Rev. A {\bf 72},
  029905 (2005).}\BibitemShut {Stop}%
\bibitem [{\citenamefont {Ng}\ \emph {et~al.}(2011)\citenamefont {Ng},
  \citenamefont {Lidar},\ and\ \citenamefont {Preskill}}]{NLP:09}%
  \BibitemOpen
  \bibfield  {author} {\bibinfo {author} {\bibfnamefont {H.~K.}\ \bibnamefont
  {Ng}}, \bibinfo {author} {\bibfnamefont {D.~A.}\ \bibnamefont {Lidar}}, \
  and\ \bibinfo {author} {\bibfnamefont {J.}~\bibnamefont {Preskill}},\ }\Doi
  {10.1103/PhysRevA.84.012305} {\bibfield  {journal} {\bibinfo  {journal}
  {Phys. Rev. A},\ }\textbf {\bibinfo {volume} {84}},\ \bibinfo {pages}
  {012305} (\bibinfo {year} {2011})}\BibitemShut {NoStop}%
\bibitem [{nor()}]{norm-to-freq}%
  \BibitemOpen
  \href@noop {} {}\bibinfo {note} {Some noise models, such as bosonic baths,
  violate the $\|H\|<\infty$ assumption. In this case our analysis still
  applies, but operator norms must be replaced by frequency cutoffs; see, e.g.,
  Refs.~\cite{NLP:09,ng:032318,Viola:98}.}\BibitemShut {Stop}%
\bibitem [{\citenamefont {Blanes}\ \emph {et~al.}(2009)\citenamefont {Blanes},
  \citenamefont {Casas}, \citenamefont {Oteo},\ and\ \citenamefont
  {Ros}}]{Blanes:08}%
  \BibitemOpen
  \bibfield  {author} {\bibinfo {author} {\bibfnamefont {S.}~\bibnamefont
  {Blanes}}, \bibinfo {author} {\bibfnamefont {F.}~\bibnamefont {Casas}},
  \bibinfo {author} {\bibfnamefont {J.}~\bibnamefont {Oteo}}, \ and\ \bibinfo
  {author} {\bibfnamefont {J.}~\bibnamefont {Ros}},\ }\href
  {http://dx.doi.org/10.1016/j.physrep.2008.11.001} {\bibfield  {journal}
  {\bibinfo  {journal} {Phys. Rep.},\ }\textbf {\bibinfo {volume} {470}},\
  \bibinfo {pages} {151} (\bibinfo {year} {2009})}\BibitemShut {NoStop}%
\bibitem [{DD-()}]{DD-review}%
  \BibitemOpen
  \href@noop {} {}\bibinfo {note} {Concatenated DD (CDD)
  \cite{KhodjastehLidar:04}, the first explicit arbitrary order DD method, uses
  a recursive nesting of elementary pulse sequences and (provided pulse
  intervals can be made arbitrarily small) can be used to achieve $N$th order
  decoupling of $n$ qubits with both $N$ and $n$ arbitrary, but requires a
  number of pulses that is exponential in both $N$ and
  $n$~\cite{KhodjastehLidar:04}. Pulse-interval optimized sequences are now
  known for purely longitudinal or purely transversal system-bath coupling,
  requiring only $N+1$ pulses for $N$th order decoupling \cite{Uhrig:07}. The
  Uhrig DD (UDD) sequence that accomplishes this was generalized to the
  quadratic DD (QDD) sequence for general decoherence of a single qubit
  \cite{WFL:09}, which uses a nesting of the transversal and longitudinal UDD
  sequences to achieve $N$th order decoupling using $(N+1)^{2}$ pulses, an
  exponential improvement over CDD and concatenated UDD \cite{CUDD}. Both UDD
  and QDD are essentially optimal in terms of the number of pulses required,
  and are provably universal for arbitrary, bounded baths
  \cite{Yang:2008:180403,Kuo:1106.2151,PhysRevA.84.060302}. Generalizing from
  QDD, nested UDD (NUDD)~pulse sequences were proposed for arbitrary
  system-environment coupling involving $n$ qubits or even higher-dimensional
  systems \cite{NUDD}. NUDD\ requires $(N+1)^{2n}$ pulses to decouple $n$
  qubits to $N$th order from an arbitrary environment.}\BibitemShut {Stop}%
\bibitem [{not()}]{notation}%
  \BibitemOpen
  \href@noop {} {}\bibinfo {note} {Throughout this work we denote the generator
  set of a set $S$ by $\hat{S}$ and the cardinality of a set $S$ by
  $|S|$.}\BibitemShut {Stop}%
\bibitem [{\citenamefont {Wocjan}(2006)}]{PhysRevA.73.062317}%
  \BibitemOpen
  \bibfield  {author} {\bibinfo {author} {\bibfnamefont {P.}~\bibnamefont
  {Wocjan}},\ }\Doi {10.1103/PhysRevA.73.062317} {\bibfield  {journal}
  {\bibinfo  {journal} {Phys. Rev. A},\ }\textbf {\bibinfo {volume} {73}},\
  \bibinfo {pages} {062317} (\bibinfo {year} {2006})}\BibitemShut {NoStop}%
\bibitem [{pro()}]{proof-th1}%
  \BibitemOpen
  \href@noop {} {}\bibinfo {note} {We prove Theorem 1. Let $B$ be generated by
  $\hat{B} = \{ b_i \}_{i=1}^{|\hat{B}|}$, so that $|B|=2^{|\hat{B}|}$, and
  consider the DD generating set $\{ \Omega_\alpha \}_{\alpha=1}^{|\hbO|}$. One
  can associate to each $b_i$ a string $s^{(i)} =
  \{s^{(i)}_1,\dots,s^{(i)}_{{|\hbO|}}\}$ where $s^{(i)}_\alpha$ encodes the
  effect the pulse $\Omega_\alpha$ has on the error term $b_i$ (commutes or
  anticommutes), via $s^{(i)}_\alpha = \pm $ if $\Omega_\alpha b_i
  \Omega_\alpha= \pm b_i$, i.e., $\Omega_\alpha b_i\Omega_\alpha =
  s^{(i)}_\alpha b_i$. The total number of such strings is $|B|$, i.e.,
  $i\in\{1,\dots,2^{|\hat{B}|}\}$. Note that if $b\in B$ is associated with the
  ``identity string" $\{+,\dots,+ \}$ then it will not be decoupled since it
  commutes with all decoupling pulses. Now, we can associate $|\hbO|$ bits
  (over the $\pm$ alphabet) to the $|\hbO|$ DD sequence generators. From these
  $|\hbO|$ bits we can construct exactly $2^{|\hbO|}$ distinct strings
  $\{r^{(i')}\}_{i'=1}^{2^{|\hbO|}}$, where
  $r^{(i')}=\{r^{(i')}_1,\dots,r^{(i')}_{{|\hbO|}}\}$, and
  $r^{(i')}_j\in\{-,+\}$. Let us map the $r^{(i')}$ strings,
  $i'\in\{1,\dots,2^{|\hbO|}\}$, to the $s^{(i)}$ strings,
  $i\in\{1,\dots,2^{|\hat{B}|}\}$. Clearly, if $B$ has ``too many" elements,
  i.e., if $|\hat{B}|>|\hbO|$, then the mapping will be one-to-many, i.e., some
  of the $r^{(i')}$ strings will have to be repeated, meaning that the set
  $\{s^{(i)}\}_{i=1}^{2^{|\hbO|}}$ will contain duplicates. The product of any
  two duplicate strings is the identity string $\{+,\dots,+ \}$. But since $B$
  is a group, this means that the product of the two distinct elements of $B$
  associated with a duplicated string is also a group element, and moreover is
  associated with the identity string. Since the elements of $B$ are in the
  Pauli group, the product of any two distinct elements cannot be the identity
  operator. Thus we have shown that there is a non-identity element of $B$
  which is associated with the identity-string, and hence is not decoupled. On
  the other hand, a DD generating set of cardinality $|\hat{B}|$ exists and is
  just $\hat{B}$ itself.}\BibitemShut {Stop}%
\bibitem [{th1()}]{th1-general}%
  \BibitemOpen
  \href@noop {} {}\bibinfo {note} {Theorem 1 can in fact be generalized by
  allowing $B$ to not be a subgroup of $P_n$, although we do not require or use
  this more general version here. The proof is similar: if a DDGS $\hbO$
  satisfying the MOOS properties exists such that the only element of $B$ that
  commutes with all elements in $\hbO$ is $\openone$ and, if each element in
  $\hat{B}$ has a unique inverse then, following an argument similar to the one
  used in \cite{proof-th1}, such a DDGS decoupling $B$ satisfies $|\hbO| \geq
  |\hat{B}|$. This more general result applies to higher dimensional
  subsystems, such as qudits. The existence of such a DDGS is guaranteed, in
  particular, for subgroups of $P_n$.}\BibitemShut {Stop}%
\bibitem [{\citenamefont {Bacon}(2006)}]{Bacon:05}%
  \BibitemOpen
  \bibfield  {author} {\bibinfo {author} {\bibfnamefont {D.}~\bibnamefont
  {Bacon}},\ }\href {http://link.aps.org/doi/10.1103/PhysRevA.73.012340}
  {\bibfield  {journal} {\bibinfo  {journal} {Phys. Rev. A},\ }\textbf
  {\bibinfo {volume} {73}},\ \bibinfo {pages} {012340} (\bibinfo {year}
  {2006})}\BibitemShut {NoStop}%
\bibitem [{\citenamefont {Kribs}\ \emph {et~al.}(2005)\citenamefont {Kribs},
  \citenamefont {Laflamme},\ and\ \citenamefont {Poulin}}]{Kribs:2005fk}%
  \BibitemOpen
  \bibfield  {author} {\bibinfo {author} {\bibfnamefont {D.}~\bibnamefont
  {Kribs}}, \bibinfo {author} {\bibfnamefont {R.}~\bibnamefont {Laflamme}}, \
  and\ \bibinfo {author} {\bibfnamefont {D.}~\bibnamefont {Poulin}},\ }\href
  {http://link.aps.org/doi/10.1103/PhysRevLett.94.180501} {\bibfield  {journal}
  {\bibinfo  {journal} {Phys. Rev. Lett.},\ }\textbf {\bibinfo {volume} {94}},\
  \bibinfo {pages} {180501} (\bibinfo {year} {2005})}\BibitemShut {NoStop}%
\bibitem [{\citenamefont {Poulin}(2005)}]{Poulin:2005uq}%
  \BibitemOpen
  \bibfield  {author} {\bibinfo {author} {\bibfnamefont {D.}~\bibnamefont
  {Poulin}},\ }\href {http://link.aps.org/doi/10.1103/PhysRevLett.95.230504}
  {\bibfield  {journal} {\bibinfo  {journal} {Phys. Rev. Lett.},\ }\textbf
  {\bibinfo {volume} {95}},\ \bibinfo {pages} {230504} (\bibinfo {year}
  {2005})}\BibitemShut {NoStop}%
\bibitem [{SDD()}]{SDD}%
  \BibitemOpen
  \href@noop {} {}\bibinfo {note} {One might try instead to choose as a DD
  sequence generator set the stabilizers only, i.e., let $\hbO= \hat{\bf S}$
  \cite{ByrdLidar:01a}. However, since the logical operators of the same code
  commute with these stabilizer DD pulses, they are not decoupled and hence
  have non-trivial action on the code subspace, thus causing logical errors.
  Formally, when $\hbO= \hat{\bf S}$, $H_0^{\textrm{eff},N}$ will contain
  logical operators.}\BibitemShut {Stop}%
\bibitem [{\citenamefont {Paz-Silva}\ \emph {et~al.}(2010)\citenamefont
  {Paz-Silva}, \citenamefont {Brennen},\ and\ \citenamefont
  {Twamley}}]{Paz-Silva:2010fk}%
  \BibitemOpen
  \bibfield  {author} {\bibinfo {author} {\bibfnamefont {G.~A.}\ \bibnamefont
  {Paz-Silva}}, \bibinfo {author} {\bibfnamefont {G.~K.}\ \bibnamefont
  {Brennen}}, \ and\ \bibinfo {author} {\bibfnamefont {J.}~\bibnamefont
  {Twamley}},\ }\href {http://link.aps.org/doi/10.1103/PhysRevLett.105.100501}
  {\bibfield  {journal} {\bibinfo  {journal} {Phys. Rev. Lett.},\ }\textbf
  {\bibinfo {volume} {105}},\ \bibinfo {pages} {100501} (\bibinfo {year}
  {2010})}\BibitemShut {NoStop}%
\bibitem [{\citenamefont {Gaitan}(2008)}]{Gaitan:book}%
  \BibitemOpen
  \bibfield  {author} {\bibinfo {author} {\bibfnamefont {F.}~\bibnamefont
  {Gaitan}},\ }\href@noop {} {\emph {\bibinfo {title} {{Quantum Error
  Correction and Fault Tolerant Quantum Computing}}}}\ (\bibinfo  {publisher}
  {{CRC}},\ \bibinfo {address} {{Boca Raton}},\ \bibinfo {year}
  {2008})\BibitemShut {NoStop}%
\bibitem [{sta()}]{stabilizer-count}%
  \BibitemOpen
  \href@noop {} {}\bibinfo {note} {The total number of physical qubits $n$ in
  an $[[n,k,d]]$ stabilizer subspace code equals the sum of the $Q=n-k$
  stabilizer and $k$ logical qubits \cite{Gottesman:96}. After concatenating
  $R$ times $n\mapsto n(R) = n^R$, $k \mapsto L(R) = k^R$, and hence $Q \mapsto
  Q(R) = n^R-k^R$. Likewise, the total number of physical qubits $n$ in an
  $[[n,k,r,d]]$ stabilizer subsystem code equals the sum of the $Q=n-(k+r)$
  stabilizer, $k$ logical, and $r$ gauge qubits \cite{Poulin:2005uq}. One can
  always view an $[[n,k,r,d]]$ subsystem code as an $[[n,k',d']]$ subspace code
  with $k'=k+r$ and distance $d'\leq d$: in a subsystem code only the $k$
  qubits designated as logical qubits are associated with the code distance
  $d$, whereas the gauge qubits have distance at most $d$. For example, in the
  $[[9,1,4,3]]$ Bacon-Shor code \cite{Bacon:05} the gauge qubits have distance
  $2$ while the logical qubit has distance $3$. Thus, after concatenating an
  $[[n,k,r,d]]$ stabilizer subsystem code $R$ times, the number of physical
  qubits is $n(R)=n^R$, which equals the sum of the $Q(R)$ stabilizer qubits,
  $L(R)=k^R$ logical qubits (with distance $d$), and $G(R)$ gauge qubits (with
  distance $\leq d$). Alternatively, viewed as an $[[n,k',d']]$ subspace code
  concatenated $R$ times, it has $L'(R)=(k')^R$ logical qubits. However, these
  logical qubits are the logical and gauge qubits of the original code, i.e.,
  $L'(R) =L(R)+G(R)$, so that $L(R)+G(R) = (k+r)^R$.}\BibitemShut {Stop}%
\bibitem [{\citenamefont {Nielsen}\ and\ \citenamefont
  {Chuang}(2000)}]{Nielsen:book}%
  \BibitemOpen
  \bibfield  {author} {\bibinfo {author} {\bibfnamefont {M.}~\bibnamefont
  {Nielsen}}\ and\ \bibinfo {author} {\bibfnamefont {I.}~\bibnamefont
  {Chuang}},\ }\href@noop {} {\emph {\bibinfo {title} {Quantum Computation and
  Quantum Information}}}\ (\bibinfo  {publisher} {Cambridge University Press},\
  \bibinfo {address} {Cambridge, England},\ \bibinfo {year} {2000})\BibitemShut
  {NoStop}%
\bibitem [{\citenamefont {Khodjasteh}\ and\ \citenamefont
  {Lidar}(2008)}]{KhodjastehLidar:08}%
  \BibitemOpen
  \bibfield  {author} {\bibinfo {author} {\bibfnamefont {K.}~\bibnamefont
  {Khodjasteh}}\ and\ \bibinfo {author} {\bibfnamefont {D.~A.}\ \bibnamefont
  {Lidar}},\ }\Doi {10.1103/PhysRevA.78.012355} {\bibfield  {journal} {\bibinfo
   {journal} {Phys. Rev. A},\ }\textbf {\bibinfo {volume} {78}},\ \bibinfo
  {pages} {012355} (\bibinfo {year} {2008})}\BibitemShut {NoStop}%
\bibitem [{\citenamefont {West}\ \emph
  {et~al.}(2010){\natexlab{b}}\citenamefont {West}, \citenamefont {Lidar},
  \citenamefont {Fong},\ and\ \citenamefont {Gyure}}]{West:10}%
  \BibitemOpen
  \bibfield  {author} {\bibinfo {author} {\bibfnamefont {J.~R.}\ \bibnamefont
  {West}}, \bibinfo {author} {\bibfnamefont {D.~A.}\ \bibnamefont {Lidar}},
  \bibinfo {author} {\bibfnamefont {B.~H.}\ \bibnamefont {Fong}}, \ and\
  \bibinfo {author} {\bibfnamefont {M.~F.}\ \bibnamefont {Gyure}},\ }\Doi
  {10.1103/PhysRevLett.105.230503} {\bibfield  {journal} {\bibinfo  {journal}
  {Phys. Rev. Lett.},\ }\textbf {\bibinfo {volume} {105}},\ \bibinfo {pages}
  {230503} (\bibinfo {year} {2010}{\natexlab{b}})}\BibitemShut {NoStop}%
\bibitem [{\citenamefont {Khodjasteh}\ \emph {et~al.}(2010)\citenamefont
  {Khodjasteh}, \citenamefont {Lidar},\ and\ \citenamefont {Viola}}]{KLV:09}%
  \BibitemOpen
  \bibfield  {author} {\bibinfo {author} {\bibfnamefont {K.}~\bibnamefont
  {Khodjasteh}}, \bibinfo {author} {\bibfnamefont {D.}~\bibnamefont {Lidar}}, \
  and\ \bibinfo {author} {\bibfnamefont {L.}~\bibnamefont {Viola}},\ }\href
  {http://link.aps.org/doi/10.1103/PhysRevLett.104.090501} {\bibfield
  {journal} {\bibinfo  {journal} {Phys. Rev. Lett.},\ }\textbf {\bibinfo
  {volume} {104}},\ \bibinfo {pages} {090501} (\bibinfo {year}
  {2010})}\BibitemShut {NoStop}%
\end{thebibliography}

%merlin.mbs 2010-03-15 4.21a (PWD, AO, DPC)

\end{document}